\documentclass[a4paper]{article}
\usepackage{amsmath,amssymb,amsthm,enumitem,times,natbib,mathtools,url,geometry,hyperref}
\usepackage[]{threeparttable}
\usepackage[utf8]{inputenc}  
\usepackage{lineno}

\newtheorem{theorem}{Theorem}
\newtheorem{remark}{Remark}

\newcommand{\T}{\intercal}

\usepackage[font=small,width=0.88\textwidth]{caption}

\def\T{\mathrm{\scriptscriptstyle T}}

\begin{document}
\title{Debiasing Welch's Method for Spectral Density Estimation}
\author{Lachlan C. Astfalck$^{1,2}$, Adam M. Sykulski$^3$, and Edward J. Cripps$^1$}
\date{}
\maketitle
\begin{abstract}
Welch's method provides an estimator of the power spectral density that is statistically consistent. This is achieved by averaging over periodograms calculated from overlapping segments of a time series.  For a finite length time series, while the variance of the estimator decreases as the number of segments increase, the magnitude of the estimator’s bias increases: a bias-variance trade-off ensues when setting the segment number. We address this issue by providing a novel method for debiasing Welch's method which maintains the computational complexity and asymptotic consistency, and leads to improved finite-sample performance. Theoretical results are given for fourth-order stationary processes with finite fourth-order moments and absolutely convergent fourth-order cumulant function. The significant bias reduction is demonstrated with numerical simulation and an application to real-world data. Our estimator also permits irregular spacing over frequency and we demonstrate how this may be employed for signal compression and further variance reduction. Code accompanying this work is available in \texttt{R} and \texttt{python}.   
\let\thefootnote\relax\footnote{$^1$ School of Physics, Mathematics \& Computing, The University of Western Australia, 35 Stirling Hwy, Crawley, Australia}
\let\thefootnote\relax\footnote{$^2$ Oceans' Graduate School, The University of Western Australia, 35 Stirling Hwy, Crawley, Australia}
\let\thefootnote\relax\footnote{$^3$ Department of Mathematics, Imperial College London, South Kensington, London SW7 2AX, U.K.}
\end{abstract}
\section{Introduction}

\graphicspath{{images/}}

The periodogram is the fundamental non-parametric estimator of the power spectral density; however, it suffers from two deficiencies: the periodogram is statistically inconsistent and is biased with finite samples. By averaging multiple periodograms calculated from partitioned or overlapping segments of a time series, \cite{welch1967use} overcomes the issue of inconsistency to provide a statistically consistent estimator of the power spectral density, eponymously known as Welch's estimator. Welch's estimator is simple to implement and easy to understand: all reasons that has led to its ubiquity in the applied sciences. However, for finite length time series, the issue of bias remains. Each segment's periodogram is biased and Welch's estimator inherits this bias. Further, as the magnitude of the bias is inversely related to the length of the segments, there is the classical inverse relationship between bias and variance of the estimator; this leads to the undesirable property that we become increasingly more certain in an estimator that is increasingly more biased. The bias is due to a phenomenon called blurring, which is observed most prominently in low power regions, and so it is often ignored as an unimportant aspect of the spectrum; although, such bias may severely impact inference and scientific understanding \citep[e.g.][]{thomson1995propagation}. In this paper, we introduce a debiased analogue to Welch's estimator that we entitle the debiased Welch estimator.

\cite{welch1967use} popularised estimating a spectral density by averaging tapered periodograms calculated from shifted overlapping windows of a time series, thus generalising Bartlett's estimator \citep{bartlett1950periodogram}. 
Statistical properties of Welch's estimator, and its behaviour under varying tapers, parametrisations, and smoothness assumptions of the underlying density, have been established; see, for instance, \cite{nuttall1971spectral}, \cite{zhurbenko1980efficiency,zurbenko1986spectral}, \cite{dahlhaus1985spectral}, and \citep{politis2005complex}. Welch's estimator may further been seen as a special case of modern methodologies such as scalable subbaging \citep{politis2024scalable}. In many cases, these developments improve upon the standard statistical properties of Welch's estimator; however, none consider the explicit debiasing under general conditions. Alternative non-parametric methodologies have been developed to mitigate the effects of bias, most prominently lag-windowing \citep{blackman1958measurement} and multi-tapering \citep{thomson1982spectrum}; although, typically these reduce correlation between distant frequencies at the cost of increasing correlation between nearby frequencies, and so bias remains. In this paper, we extend the results in \cite{sykulski2019debiased}, who provided methodology to debias parametric spectral density estimates, to non-parametric estimation by constructing a natural basis representation to the power spectral density.
From this basis representation we intentionally bias the bases, fit the biased bases to a biased estimate of the power spectral density, and recover debiased coefficients that are used to form our debiased non-parametric estimator.

Novel insight is provided across two theorems. The first presents theoretical asymptotic properties of Welch's estimator, making the finite sample bias explicit and motivating the construction of a new debiased alternative which we propose via a basis representation. To obtain a good representation of the power spectral density, we must specify a large number of basis functions and so estimation of the coefficients via direct optimisation of the Whittle likelihood, as in \cite{sykulski2019debiased}, becomes computationally infeasible. The major contribution of this work is therefore housed in the second theorem, where we introduce methodology that leverages results established in the first theorem, to provide asymptotic results that enable fast and direct computation of the debiased Welch estimator. We build theory for time series requiring fourth-order stationarity with finite fourth-order moments and absolutely convergent fourth-order cumulant function. The theory is established for all tapers, subject to the standard conditions of squared summability of the tapering sequence, thereby extending many classical results that use the raw periodogram, for example, those in \cite{anderson1971statistical}, \cite{brillinger2001time} and \cite{katznelson1968introduction}.

The theoretical developments are also demonstrated by simulation. To evidence the asymptotics, we use a canonical auto-regressive model of order four studied throughout the spectral estimation literature \citep[see for example][]{percival1993spectral}; then, we present results for shorter-length time series generated from a Mat\'ern covariance function. 
Next, we demonstrate the debiased Welch estimator in a real-world case study of coastal wave measurements. This provides an example of where parametric methods struggle to capture the complexity of the process, thus necessitating the use of non-parametric methods. Finally, unlike Welch's estimator, the debiased Welch estimator does not require even spacing over frequency; we use this to demonstrate how we may simultaneously debias and compress the signal. 
Code accompanies our results; \texttt{R} and \texttt{python} implementations are available at \texttt{github.com/TIDE-ITRH}.

\section{Spectral analysis} \label{sec:spec_anal}

\subsection{Fundamentals and definitions}

Define by $\{X_t\}$ a discrete real-valued stochastic process, sampled at interval $\Delta$ and indexed by $t \in \mathcal{Z}$, where $\mathcal{Z}$ is the set of integers. The process $\{X_t\}$ may be either thought of as a true discrete-time stochastic process, such as an auto-regressive model, or as uniformly sampled observations of a continuous-time stochastic process $\{X(t)\}$ in which case $X_t = X(\Delta t)$. Define the auto-covariance sequence of $\{X_t\}$ as $\gamma(\tau) = \mathrm{E}[X_t X_{t-\tau}]$ for $\tau \in \mathcal{Z}$. Following Wiener–Khinchin's theorem, a necessary and sufficient condition for $\gamma(\tau)$ to be the auto-covariance function of $\{X_t\}$ is that there exists a power spectral distribution function $F(\omega)$ such that
\begin{equation} \label{eqn:riemann-stieltjes}
  \gamma(\tau) = \frac{1}{2\pi}\int_{-\pi/\Delta}^{\pi/\Delta} e^{i \omega \tau \Delta} \; \mathrm{d} F(\omega),
\end{equation}
for $\omega \in [-\pi/\Delta, \pi/\Delta]$ defined in radians, where \eqref{eqn:riemann-stieltjes} is a Riemann-Stieltjes integral. We assume fourth-order stationarity, with finite fourth-order moments and fourth-order cumulant function
\begin{equation*}
  \mathcal{K}(\tau_i, \tau_j, \tau_k) = \mathrm{E}[X_0 X_{\tau_i} X_{\tau_j} X_{\tau_k}] - \gamma(\tau_i)\gamma(\tau_j - \tau_k) - \gamma(\tau_j)\gamma(\tau_i - \tau_k) - \gamma(\tau_k)\gamma(\tau_i - \tau_j)
\end{equation*}
that is absolutely convergent. As $\sum_{\tau_i, \tau_j = -\infty}^{\infty} |\gamma(\tau_i)\gamma(\tau_j)|$ is convergent only if $\sum_{\tau = \infty}^{\infty} |\gamma(\tau)| < \infty$, then $\gamma(\tau)$ too is absolutely convergent.
Consequently, the power spectral distribution is absolutely continuous \citep[see Section~4.2][]{shumway2000time}, and so there exists a power spectral density $f(\omega)$ that satisfies $ f(\omega) \mathrm{d}\omega = \mathrm{d} F(\omega)$, and $f(\omega)$ forms a Fourier pair with $\gamma(\tau)$ so that 
\begin{equation} \label{eqn:fourier}
    \gamma(\tau) = \frac{1}{2\pi}\int_{-\pi/\Delta}^{\pi/\Delta} f(\omega) e^{i \omega \tau \Delta} \; \mathrm{d} \omega, \hspace{3mm} f(\omega) = \Delta \sum_{\tau = -\infty}^{\infty} \gamma(\tau) e^{-i \omega \tau \Delta}.
\end{equation}
As $f(\omega)$ is uniformly continuous over this interval, and according to the Vitali convergence theorem, $f(\omega)$ is Riemann integrable in $\omega$.

Say we observe from $\{X_t\}$ an $n$-dimensional observation $\mathrm{X}_n = (x_0, \dots, x_{n-1})$. The most common estimator of $f(\omega)$ from $\mathrm{X}_n$ is the periodogram, $I_n(\omega)$, defined as the squared modulus of the discrete Fourier transform of $\mathrm{X}_n$,
\begin{equation} \label{eqn:periodogram}
    I_n(\omega) = |J_n(\omega)|^2, \hspace{3mm} J_n(\omega) = \Delta^{1/2} n^{-1/2} \sum_{t=0}^{n-1}  x_t e^{- i \omega t \Delta}.
\end{equation}
Alternatively, we may calculate the periodogram via the discrete Fourier transform of $\check{\gamma}(\tau)$, a biased estimator to $\gamma(\tau)$,
\begin{equation*}
  I_n(\omega) = \Delta \sum_{\tau = -n+1}^{n-1} \check{\gamma}(\tau) e^{-i \omega \tau \Delta}, \hspace{3mm} \check{\gamma}(\tau) = \frac{1}{n} \sum_{t = 0}^{n - |\tau| -1} x_{t + |\tau|} x_t.
\end{equation*}
The periodogram is a biased and inconsistent estimator of $f(\omega)$. It has expectation $\mathrm{E}[I_n(\omega)] = f(\omega) * \mathcal{F}_n(\omega)$ where 
\begin{equation} \label{eqn:fejer}
\mathcal{F}_n(\omega) = \frac{\Delta}{2 \pi n} \frac{\sin^2(n \omega \Delta/2)}{\sin^2(\omega \Delta/2)}
\end{equation}
is the Fej\'er kernel, and variance $\mathrm{var}[I_n(\omega)] \approx f^2(\omega)$ which does not reduce with $n$, thus causing the inconsistency. The periodogram is asymptotically unbiased as $\mathcal{F}_n(\omega)$ tends to a Dirac delta function as $n \rightarrow \infty$. Often, practitioners will appeal to this to argue that the periodogram is a suitable estimator; however, the effects due to bias are related not only to $n$, but also to the dynamic range of $f(\omega)$. Even with massive $n$, for common $f(\omega)$ encountered in the sciences, the periodogram can be too biased to be useful \citep[e.g. see the discussion in][]{thomson1982spectrum}.

A standard technique to alleviate the bias in the periodogram is to apply a taper, or window, to the data. Tapering pre-multiplies the data with a taper to form the product $h_t X_t$ for $t = 0, \dots, n-1$ where $\sum_{t=0}^{n-1} h_t^2 = 1$. The sequence $\{h_t\}$ is called the data taper. The tapered periodogram, otherwise known as the modified periodogram, is defined as 
\begin{equation*}
    I_n(\omega; h) = |J_n(\omega; h)|^2, \hspace{3mm} J_n(\omega; h) = \Delta^{1/2} \sum_{t=0}^{n-1} h_t x_t e^{- i \omega t \Delta}.
\end{equation*}
The expectation of the tapered periodogram is $\mathrm{E}[I_n(\omega; h)] = f(\omega) * \mathcal{H}_n(\omega)$ where $\mathcal{H}_n(\omega)$ is the spectral window for $\{h_t\}$ defined by
\begin{equation*}
\mathcal{H}_n(\omega) = \Delta \left| \sum_{t=0}^{n-1} h_t e^{- i \omega t \Delta} \right|^2.
\end{equation*}
Tapering can be effective in removing bias; however, it is not a panacea. The data taper is generally chosen so as to reduce bias in the form of broadband blurring in $I_n(\omega; h)$; this comes with the trade-off that bias from narrowband blurring is increased. Further, tapering, in effect, reduces the amount of data and so the tapered periodogram has a higher variance as compared to the standard periodogram. This variance can be reduced by applying multiple tapers (multi-tapering), but this then reintroduces bias such that the level of the bias-variance trade-off must be in practice controlled by the practitioner. Setting $h_t = n^{-1/2}$ recovers the periodogram as $\mathcal{H}_n(\omega) = \mathcal{F}_n(\omega)$. For this reason all periodograms may be considered to be tapered, and so we will use the notation $I_n(\omega; h)$ to refer to both tapered and standard periodograms throughout.

\subsection{Debiased parametric estimation}

Rather than using $I_n(\omega; h)$ as a direct estimate to $f(\omega)$, it is common to instead use $I_n(\omega; h)$ to fit some parametric form $f(\omega; \vartheta)$, with Fourier pair $\gamma(\omega; \vartheta)$. To avoid the unwieldy computation often associated with time-domain methods requiring matrix inversion, it is standard to approximate the exact maximum likelihood with a discretized form of the Whittle likelihood
\begin{equation} \label{eqn:whittle}
l_\mathrm{W}(\vartheta) = - \sum_{\omega \in \Omega_n} \left( \log f(\omega; \vartheta) + \frac{I_n(\omega; h)}{f(\omega; \vartheta)}\right),
\end{equation}
where $\Omega_n = 2\pi(\Delta n)^{-1}(-\lfloor n/2 \rfloor, \dots, -1, 0, 1, \dots, \lceil n/2 \rceil - 1)$ is the set of Fourier frequencies for an observed stochastic process of length $n$. Due to the inherent bias in $I_n(\omega; h)$, parameter estimation from using \eqref{eqn:whittle} can be biased \citep[for example, see][]{dahlhaus1988small}. \citet{sykulski2019debiased} proposed a solution to this problem, whereby the pseudolikelihood
\begin{equation} \label{eqn:debiased_whittle}
l_\mathrm{D}(\vartheta) = - \sum_{\omega \in \Omega_n} \left( \log \check{f}_n(\omega; \vartheta) + \frac{I_n(\omega; h)}{\check{f}_n(\omega; \vartheta)}\right)
\end{equation}
is defined. Here, we substitute $\check{f}_n(\omega; \vartheta) = E[I_n(\omega; h)] = f(\omega; \vartheta) * \mathcal{H}_n(\omega)$ for $f(\omega; \vartheta)$ in \eqref{eqn:whittle}. We use the caron diacritic throughout this manuscript to indicate intentionally biased quantities. \citet{sykulski2019debiased} further showed that the convolution $f(\omega; \vartheta) * \mathcal{H}_n(\omega)$ can be efficiently calculated as
\begin{equation} \label{eqn:debiased_fft}
\check{f}_n(\omega; \vartheta) = 2\Delta \times \mathrm{Re}\left\{\sum_{\tau = 0}^{n-1} \left( \sum_{t = 0}^{n - \tau -1} h_t h_{t+\tau} \right) \gamma(\tau; \vartheta) e^{-i\omega \tau \Delta}\right\} - \Delta \times \gamma(0; \vartheta)
\end{equation}
using fast Fourier transforms. When $h_t = n^{-1/2}$, and so $I_n(\omega; h)$ is the periodogram, $\sum_{t = 0}^{n - \tau -1} h_t h_{t+\tau} = 1 - \tau/n$, and \eqref{eqn:debiased_fft} is the standard form for the debiased Whittle likelihood as presented in Equation~(9) of \cite{sykulski2019debiased}. This technique debiases estimates of $\vartheta$, and hence $f(\omega; \vartheta)$ and $\gamma(\omega; \vartheta)$, in parametric estimation.

\subsection{Welch's method for spectral density estimation}

Welch's method segments a time series $x_0, \dots, x_{n-1}$ into $M$ partitioned or overlapping blocks of length $L$, and applies a length $L$ data taper to each block. We define the periodogram of the $m$th block, for $m = 0, \dots, M-1$, as
\begin{equation*}
I_{L}^{m}(\omega; h) = \Delta \left| \sum_{t=0}^{L-1} h_t x_{(t + mS)} e^{-i\omega t \Delta} \right|^2,
\end{equation*}
where $S$ is an integer-valued shift factor satisfying $0 < S \leq L$, so that $p = (L-S)/L$, where $0 \leq p < 1$, is the percentage overlap of the blocks, and the relationship between $M$, $L$ and $p$ is
\begin{equation} \label{eqn:ml_tradeoff}
    L[M(1-p)+p] = n,
\end{equation}
such that setting two of $\{M,L,p\}$ fixes the third for a given $n$.
Welch's estimator is defined as
\begin{equation} \label{eqn:welch_estimator}
  \bar{I}_L(\omega; h) = \frac{1}{M}\sum_{m = 0}^{M-1}I_L^{m}(\omega; h).
\end{equation}
The degree of overlap, $p$, and the choice of data taper, $\{h_t\}$, are modelling choices often informed by either desired properties of the analysis or preferences within scientific disciplines. For instance, it is standard in the engineering sciences to assume a 50\% overlap ($p=0.5$) with a Hamming taper: see the \texttt{MATLAB} documentation of the function \texttt{pwelch}. Welch's estimator has expectation $\mathrm{E}[\bar{I}_L(\omega; h)] = f(\omega) * \mathcal{H}_L(\omega)$ and variance
\begin{align}
\mathrm{var}\left[\bar{I}_{L}(\omega; h)\right] =& \frac{1}{M^2} \sum_{m=0}^{M-1} \mathrm{var}\left[I_{L}^{m}(\omega; h)\right] + \frac{2}{M^2} \sum_{m<m'} \mathrm{cov}\left[I_{L}^{m}(\omega; h), I_{L}^{m'}(\omega; h)\right] \nonumber\\
=& \frac{\mathrm{var}[I_{L}(\omega; h)]}{M}\left(1 + \frac{2}{M} \sum_{m<m'} \rho(\omega; |m-m'|, h) \right), \label{eqn:welch_variance2}
\end{align}
where $\rho(\omega; |m-m'|, h) = \mathrm{cor}\left[I_{L}^{m}(\omega; h), I_{L}^{m'}(\omega; h)\right]$ and $\mathrm{var}[I_{L}(\omega; h)] = \mathrm{var}[I_{L}^{m}(\omega; h)]$ for all $m$. When $n >> M$, Welch's estimator is approximately equal to a lag-window estimator where the lag-window is given by the convolution of $h_t$ with itself. This is most simple to see for Bartlett's estimator ($p = 0$ and $h_t = L^{-1/2}$) where the lag-window is a triangular function of width $2L$, which results from the discrete convolution of $h_t = L^{-1/2}$. The primary focus of this manuscript is to study the debiasing of Welch's estimator; although, we note, that the presented results may also be readily extended to certain lag-window estimators. We now establish some important properties of Welch's estimator in Theorem~\ref{the:welch_var}.

\begin{theorem} \label{the:welch_var}
  Assume a fourth-order stationary stochastic process $\{X_t\}$ with finite fourth-order moments and absolutely convergent fourth-order cumulant function. Welch's estimator, defined in \eqref{eqn:welch_estimator}, of the observed process $\mathrm{X}_n = (x_0, \dots, x_{n-1})$, with data taper $\{h_t\}$ that satisfies $\sum h_t^2 = 1$, has the following properties.
  \begin{enumerate}[label=\alph*)]
    \item The estimator, $\bar{I}_{L}(\omega; h)$, satisfies $\bar{I}_{L}(\omega; h) = f(\omega) * \mathcal{H}_L(\omega) + \mathcal{O}_p(M^{-1/2})$.
    \item The variance $\mathrm{var}[\bar{I}_L(\omega;h)]$ converges asymptotically, in $L$, to 
    \begin{equation*}
    \lim_{L \rightarrow \infty} \mathrm{var}[\bar{I}_L(\omega;h)] = c \; \mathrm{var}[I_L(\omega; h)]
    \end{equation*}
    for all $\omega \neq 0, \omega_\mathcal{N}$. Here, $c$ is a constant of proportionality that is constant over $\omega$, and $\omega_\mathcal{N}$ denotes the Nyquist frequency.
    \item The quantity $M^{-1/2} \bar{I}_{L}(\omega; h)$ tends to a Gaussian distribution as $M \rightarrow \infty$.
  \end{enumerate}
\end{theorem}
\begin{proof}
Proofs are provided in the Supplementary Material.
\end{proof}
\begin{remark}
  As $n$ increases, then it is not possible to increase both $L$, the block length, and $M$, the number of blocks, linearly with $n$ for a fixed overlap $p$, see \eqref{eqn:ml_tradeoff}; a trade-off is required. For fixed $L$ and linearly increasing $M$, Theorem~\ref{the:welch_var}a establishes that Welch's estimate has the optimal convergence rate of $\mathcal{O}_p(M^{-1/2}) = \mathcal{O}_p(n^{-1/2})$, but the estimator is inconsistent for all processes except white noise due to the bias inherent in $f(\omega) * \mathcal{H}_L(\omega)$ for fixed $L$ as discussed. On the other extreme, for fixed $M$ and linearly increasing $L$, the bias asymptotically vanishes but the variance is $\mathcal{O}_p(1)$ (Theorem~\ref{the:welch_var}a) and therefore the Welch estimate suffers the same inconsistency as the periodogram (Theorem~\ref{the:welch_var}b). Consistency is only achieved by increasing both $L$ and $M$ with $n$ at reduced rates, say, $L = \mathcal{O}(n^\alpha)$ and $M = \mathcal{O}(n^{1 - \alpha})$ for $0 < \alpha < 1$, to satisfy \eqref{eqn:ml_tradeoff}. In practice, regularity assumptions may be employed to aid the selection of $L$ and $M$. For example, if we assume $h_t = L^{-1/2}$ and $f(\omega)$ to be Lipschitz continuous then, as shown in \cite{fejer1910lebesguessche},
  \begin{equation*}
    \mathrm{bias}\left[\bar{I}_{L}(\omega; h)\right] = \mathrm{E}[\bar{I}_{L}(\omega; h)] - f(\omega) = \mathcal{O}_p\left(\frac{\log L}{L}\right)
  \end{equation*}
  and so the optimal value of $\alpha$, with respect to mean squared error, is approximately $\alpha = 1/3$. Further regularity conditions and specification of tapers may yield even faster convergence of the bias: assumptions on the order of differentiability of $f(\omega)$ can further reduce the optimal choice of $\alpha$ \citep[see, for example,][]{zhurbenko1980efficiency}, and tapers may be chosen with associated properties, such as the family of Slepian tapers that maximally concentrate bandwidth energy \citep{slepian1978prolate}.
\end{remark}

\section{Debiasing Welch's Method} \label{sec:debiased_welch}

\subsection{Constructing biased bases} \label{sec:biased_bases}

To construct a debiased and consistent estimate to $f(\omega)$ based on Welch's method, we use concepts proposed by \cite{sykulski2019debiased} but extend them to a non-parametric setting. Specifically, 
here we construct from a family of bases $B_k(\omega)$, a family of biased bases $\check{B}_k(\omega)$. We assume the functional form $f(\omega; \vartheta) = \sum_{k=1}^K a_k B_k(\omega)$ and set $\check{f}(\omega; \vartheta) = \sum_{k=1}^K a_k \check{B}_k(\omega)$ in \eqref{eqn:debiased_whittle} to provide debiased inference for the parameters $\vartheta = (a_1, \dots, a_K)$. Here, the $B_k(\omega)$ are basis functions specified from some family of bases chosen to represent some desirable property of $f(\omega)$. We set $\int_{-\infty}^{\infty} B_k(\omega) \mathrm{d} \omega = 1$. Due to the linearity of the integral in \eqref{eqn:fourier}, assuming $f(\omega; \vartheta) = \sum_{k = 1}^K a_k B_k(\omega)$ implies $\gamma(\tau; \vartheta) = \sum_{k = 1}^K a_k \rho_k(\tau)$ where each $\rho_k(\tau)$ are the auto-correlation functions calculated from $B_k(\omega)$ as
\begin{equation} \label{eqn:gamma_i}
    \rho_k(\tau) = \frac{1}{2\pi} \int_{-\pi/\Delta}^{\pi/\Delta} B_k(\omega) e^{i \omega \tau \Delta} \mathrm{d}\omega.
\end{equation}
Due to the distributivity of the convolution operator the expected periodogram is $\check{f}(\omega; \vartheta, h) = \mathrm{E}[I_n(\omega; h)] = \sum_{k=1}^K a_k \check{B}_k(\omega; h)$ where each $\check{B}_k(\omega; h)$ is calculated exactly using a fast Fourier transform by substituting $\rho_k(\tau)$ for $\gamma(\tau; \vartheta)$ in \eqref{eqn:debiased_fft}.

\subsection{A Riemann approximation to the spectral density} \label{sec:riemann_approx}

Section~\ref{sec:biased_bases} describes a general approach to bias any basis family. Here, we provide specific details with respect to a Riemannian approximation of $f(\omega)$. As $f(\omega)$ is Riemann integrable then we may approximate $f(\omega)$ by summating contiguous rectangular functions; this approximation becomes exact as the width of the rectangular functions tends to zero. Define the rectangular function by
\begin{equation*}
  \mathrm{rect}(\omega) = \begin{cases} 1 & |\omega| < 1/2 \\ 1/2 & |\omega| = 1/2 \\ 0 & \text{elsewhere} \end{cases}
\end{equation*}
and the symmetric rectangular function with centre $\omega^\mathrm{c}$ and width $\delta$ as
\begin{equation} \label{eqn:sym_rect}
  \mathrm{symrect}(\omega; \omega^\mathrm{c}, \delta) = \mathrm{rect}\left(\frac{\omega - \omega^\mathrm{c}}{\delta}\right) + \mathrm{rect}\left(\frac{\omega + \omega^\mathrm{c}}{\delta}\right),
\end{equation}
for $\omega^\mathrm{c} \geq \delta/2$. As we are concerned with real-valued time series, $f(\omega) = f(-\omega)$. 
Define $(\omega_0, \omega_1, \dots, \omega_K)$ as a partition of $[0, \omega_\mathcal{N}]$. The mid-point Riemann approximation to $f(\omega)$ is
\begin{equation} \label{eqn:riemann}
  f(\omega) \approx \sum_{k=1}^K f\left(\omega_k^\mathrm{c}\right) \mathrm{symrect}\left(\omega; \omega_k^\mathrm{c}, \delta_k \right), \quad \omega \in [-\pi/\Delta, \pi/\Delta],
\end{equation}
with centres $\omega_k^\mathrm{c} = \omega_{k-1} + \delta_k/2$ and widths $\delta_k = \omega_{k} - \omega_{k-1}$. The approximation in \eqref{eqn:riemann} converges to equality as $K \rightarrow \infty$ and all $\delta_k \rightarrow 0$. This approach does not require us to space the bases regularly; any partition can be defined, allowing for irregularly spaced bases. For now, so as to mimic the behaviour of Welch's estimate, we assume the bases to be regularly spaced in $\omega$, and so $\delta_k = \delta$ for all $k$. We define the centres and widths of the symmetric rectangular functions so that for a maximum value of $K = \lceil (L - 1)/2 \rceil$, $\omega_k^\mathrm{c} = k \delta - \delta/2$ and $\delta = \pi /(\Delta K)$. Setting $B_k(\omega) = \mathrm{symrect}\left(\omega; \omega_k^\mathrm{c}, \delta \right)$, and following the results of \citet{tobar2019band}, we calculate via \eqref{eqn:gamma_i}
\begin{equation} \label{eqn:sinc_kernel}
  \rho_k\left(\tau; \omega_k^\mathrm{c}, \delta \right) =  \delta^{-1}\mathrm{sinc}\left(\tau \delta \right) \cos\left(\omega_k^\mathrm{c} \tau \right),
\end{equation}
where $\mathrm{sinc}(\tau) = (\pi \tau)^{-1}\sin(\pi \tau)$ and $\mathrm{sinc}(0) = 1$. To calculate $\check{B}_k(\omega; h)$ we substitute $\rho_k\left(\tau; \omega_k^\mathrm{c}, \delta \right)$ in \eqref{eqn:debiased_fft}.  

\subsection{Model fitting} \label{sec:model_fitting}

For a given choice of bases $B_k(\omega) $ the final step is to estimate $\vartheta=(a_1,...,a_k)$ to recover the spectral estimate $f(\omega; \vartheta) = \sum_{k=1}^K a_k B_k(\omega)$. The preferred debiased approach would be to specify $\check{f}(\omega; \vartheta) = \sum_{k = 1}^K a_k \check{B}_k(\omega; h)$ and solve for the $a_k$ by maximising the debiased Whittle likelihood defined in \eqref{eqn:debiased_whittle}. For the approximation to $f(\omega)$ in \eqref{eqn:riemann} to be reasonable, however, we are required to define a potentially large number of bases, $K$, and so it is likely that maximising \eqref{eqn:debiased_whittle} is computationally infeasible using naive optimisation routines. 

To proceed we observe from Theorem~\ref{the:welch_var}c that $\bar{I}_L(\omega;h)$ converges, with $M$, to be Gaussian, and so we model $\bar{I}_L(\omega; h) = \sum_{k=1}^K a_k \check{B}_k(\omega; h) + \epsilon(\omega)$, for heteroskedastic and Gaussian $\epsilon(\omega)$ with $\mathrm{var}[\epsilon(\omega)] = \mathrm{var}[\bar{I}_L(\omega; h)]$. Define the vector $\check{\mathrm{B}}(\omega; h) = (\check{B}_1(\omega; h), \dots, \check{B}_K(\omega; h))$ as the column vector of all the biased bases at $\omega$. Weighted least squares can be used to estimate a solution of the form
\begin{equation} \label{eqn:least_squares}
  \underset{\vartheta}{\arg \min} \; \sum_{\omega \in \Omega_L} \left\{\mathrm{var}[\bar{I}_L(\omega; h)]^{-1} \left(\bar{I}_L(\omega; h) - \vartheta^\T \check{\mathrm{B}}(\omega; h) \right)^2\right\}.
\end{equation}
As $\mathrm{var}[\bar{I}_L(\omega; h)]$ is dependent on $f(\omega)$, the true power spectral density we are trying to estimate, the solution to \eqref{eqn:least_squares} is not analytical. As an approximation to \eqref{eqn:least_squares}, we substitute $\bar{I}_L(\omega;h)^2$ for $\mathrm{var}[\bar{I}_L(\omega; h)]$ and instead solve 
\begin{equation} \label{eqn:least_squares2}
  \hat{\vartheta} = \underset{\vartheta}{\arg \min} \; \sum_{\omega \in \Omega_L} \left\{\bar{I}_L(\omega;h)^{-2} \left(\bar{I}_L(\omega; h) - \vartheta^\T \check{\mathrm{B}}(\omega; h) \right)^2\right\}.
\end{equation}
The motivation for \eqref{eqn:least_squares2} is as follows. From Theorem~\ref{the:welch_var}b, as $\lim_{L\rightarrow \infty} \mathrm{var}\left[\bar{I}_{L}(\omega; h)\right] = c \; \mathrm{var}\left[I_{L}(\omega; h)\right]$, over all $\omega$, we substitute $\mathrm{var}\left[I_{L}(\omega; h)\right]$ for $\mathrm{var}\left[\bar{I}_{L}(\omega; h)\right]$ in \eqref{eqn:least_squares}. Further, as we establish in Theorem~\ref{the:variance} below, $\bar{I}_L(\omega;h)^2$ converges in $M$ and $L$ to $\mathrm{var}[I_L(\omega; h)]$, and so we approximate $\mathrm{var}[I_L(\omega; h)]$ by $\bar{I}_L(\omega;h)^2$, resulting in \eqref{eqn:least_squares2}. From $\hat{\vartheta} = (\hat{a}_1, \dots, \hat{a}_K)$, we define a debiased functional estimate of $f(\omega)$ as $f(\omega; \hat{\vartheta}) = \sum_{k = 1}^{K} \hat{a}_k B_k(\omega)$. For the symmetric rectangular basis function of Section 3.2, we define the debiased Welch estimator at a discrete set of frequencies as 
\begin{equation*}
    \hat{I}(\omega_k) = \hat{a}_k B_k(\omega_k) = \hat{a}_k / \delta_k, \quad \omega_k \in \{\omega_1^\mathrm{c}, \dots, \omega_K^\mathrm{c}\};
\end{equation*}
this is analogous to the frequencies for which Welch's estimator is generally defined. Due to finite-sample error, the debiased Welch estimator is not guaranteed to yield a non-negative solution. We address this in Section~\ref{sec:computation} by constraining the solution space of $\hat{\vartheta}$ so that $\hat{\vartheta} \in \mathcal{R}^+$ where $\mathcal{R}^+$ denotes the non-negative space of reals. As by definition, $f(\omega) \in \mathcal{R}^+$, and following similar arguments to Theorem~2 of \cite{politis1995bias}, we can establish that constraining $\hat{\vartheta} \in \mathcal{R}^+$ yields a better estimator with respect to mean squared error.

\begin{theorem} \label{the:variance}
Assume a fourth-order stationary process, $\{X_t\}$, with finite fourth-order moments and absolutely convergent fourth-order cumulant function, $\mathcal{K}(\tau_i, \tau_j, \tau_k)$. The variance of the periodogram of any observed length-$L$ segment of $\mathrm{X}_n = (x_0, \dots, x_{n-1})$, for all $\omega \neq 0, \omega_\mathcal{N}$, and with data taper $\{h_t\}$ that satisfies $\sum h_t^2 = 1$, is
\begin{enumerate}[label=\alph*)]
\item $\mathrm{var}[I_L(\omega; h)] = \bar{I}_L(\omega;h)^2 + \mathcal{O}_p \left(\frac{1}{\sqrt{M}} + \frac{\log^2 L}{L^2}\right)$, when $\mathcal{K}(\tau_i, \tau_j, \tau_k) = 0$ for all $\tau_i, \tau_j, \tau_k$; and
\item $\mathrm{var}[I_L(\omega; h)] = \bar{I}_L(\omega;h)^2 + \mathcal{O}_p \left(\frac{1}{\sqrt{M}} + \frac{1}{L}\right)$, otherwise.
\end{enumerate}
\end{theorem}
\begin{proof}
The proof is provided in the Supplementary Material.
\end{proof} 

\begin{remark}
    From Theorem~1 of \cite{sykulski2019debiased}, we find that the debiased Welch estimator, $\hat{I}(\omega_k)$, when computed through \eqref{eqn:debiased_whittle}, is unbiased so that $\mathrm{E}[\hat{I}(\omega_k)] = f(\omega_k; \vartheta)$. The bias that results from the Riemannian basis approximation can be readily quantified. As $f(\omega)$ has a bounded first derivative $f'(\omega)$, the bias that results from the Riemann approximation is $\mathrm{bias}[f(\omega; \vartheta)] = \mathcal{O}_p(L^{-1})$. Following \cite{chui1971concerning}, if we further assume that $f'(\omega)$ is of bounded variation, we may strengthen this statement to say $\mathrm{bias}[f(\omega; \vartheta)] = \mathcal{O}_p(L^{-2})$. As $\hat{I}(\omega_k)$ is a linear transform of $\bar{I}_L(\omega;h)$ we find that $\mathrm{var}[\hat{I}(\omega_k)] = C(K,L,h)\mathrm{var}[\bar{I}_L(\omega;h)] = \mathcal{O}_p(M^{-1})$ where $C(K,L,h) > 0$ is a constant that depends on the choice of $K$, $L$ and the data taper, and is not necessarily bounded above by 1. Thus, in terms of mean squared error, the optimal parametrisation is $L = \mathcal{O}(n^{1/3})$, or $L = \mathcal{O}(n^{1/5})$ if $f'(\omega)$ is of bounded variation. The weighted least squares routine in \eqref{eqn:least_squares2} follows asymptotic arguments in Theorems~1 and 2 that result in an estimator that is asymptotically unbiased and $\mathcal{O}_p(M^{-1/2})$ as long as $L$ and $M$ both increase with $n$. As $C(K,L,h)$ is not bounded above by $1$, the debiased Welch estimator is not guaranteed to have a smaller mean squared error as compared to the Welch estimator given the same parameterisation. However, as the debiased Welch estimator corrects for the effects of bias, in practice we may apportion between $L = \mathcal{O}(n^\alpha)$ and $M = \mathcal{O}(n^{1-\alpha})$ with a lower $\alpha$ than would be chosen for Welch's estimator to minimise mean squared error (see also Remark 1), allowing for extra variance reduction, as we shall demonstrate for finite samples in Section~\ref{sec:matern}. This is a similar result to \cite{politis1995bias} who study bias-correction of lag-window estimators.
\end{remark}

\subsection{Computation} \label{sec:computation}

We summarize the computation of the debiased Welch estimator in Algorithm~1 in the Supplementary Material, with two variants provided. The first variant, shown in lines 3--8 of Algorithm~1, provides the option to space the $K$ bases evenly or unevenly over frequency. Even spacing of the bases provides a closer analogue to Welch's estimator; here, the choice of integer valued $K \leq \lceil (L-1)/2 \rceil$ is the only additional user decision required by the debiased Welch estimator. When there is large bias present in Welch's estimator, choosing $K$ close to $\lceil (L-1)/2 \rceil$ can lead to an ill-defined solution. In general, we recommend $K = \lceil (L-1)/4 \rceil$ as a good initial estimate which down-samples the Fourier frequencies by an order of two. For uneven spacing, the centres and widths are all user-specified. We discuss uneven spacing further in Section~6 with an example provided for high-frequency slope estimation. The second variant of the model, shown in lines 12--15 of Algorithm~1, provides an option to impose non-negativity on $\vartheta$ in \eqref{eqn:least_squares2}, constraining the solution space to $\hat{\vartheta} \in \mathcal{R}^+$. As discussed, the debiased Welch estimate may result in small negative values. This is particularly so when the choice of $K$ is too high and so the solution to the linear system in \eqref{eqn:least_squares2} is unstable. Efficiently computing non-negative least-squares estimates has been long understood; in the accompanying code we implement results from \cite{lawson1995solving}. Note, that the non-negative least-squares solution is guaranteed to have a lower mean squared error, see proof of Theorem 2 in \cite{politis1995bias}. The computational order of the debiased Welch estimator, for both even and uneven spacing, is $\mathcal{O}(\max\{n \log L, K^3\})$. Here, the first term is the computational order of Welch's estimator, and the second term is the computational order to do the debiasing in \eqref{eqn:least_squares2}. The number of bases $K$ scales with $L$ and so $\mathcal{O}(K^3) \propto \mathcal{O}(L^3)$ for even spacing. From Theorem~\ref{the:variance}, we observe that the optimal rate with which to scale $L$ with $n$ is at maximum $\mathcal{O}(n^{1/3})$, and so the debiasing in \eqref{eqn:least_squares2} scales at most $\mathcal{O}(n)$. Thus, the debiased Welch estimator retains the same computational order of the standard Welch estimator, $\mathcal{O}(n \log L)$.

\section{Simulations studies} \label{sec:simulations}

\subsection{Autoregressive processes} \label{sec:ar4}

\begin{figure}[b!]
  \centering
  \includegraphics[width = 0.7\linewidth]{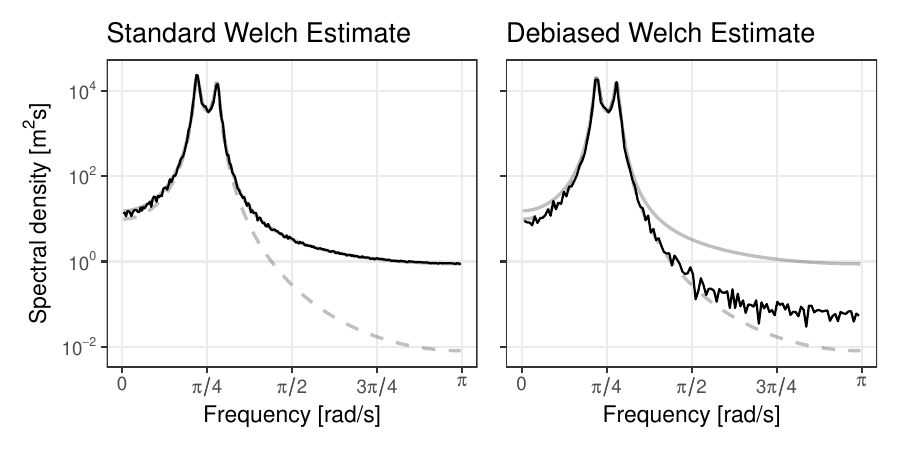}
  \caption{Welch estimates of an $\textsc{ar(4)}$ random time series. The $\textsc{ar(4)}$ model is defined formally in the text; the true spectral density is shown by the dashed grey line and the expected Welch estimate is shown by the solid grey line. The Welch estimates are shown by solid black lines with the standard Welch estimate shown left, and the debiased Welch estimate shown right.}
  \label{fig:ar4_peridograms}
\end{figure}

To test empirical performance of the debiased Welch estimator, we run a simulation study using a canonical model for biased spectral density estimation: the $\textsc{ar(4)}$ model of \cite{percival1993spectral}. The model is specified as $X_t = \sum_{j = 1}^4 \phi_j X_{t - j} + \epsilon_t$ where $\epsilon_t \sim \mathcal{N}(0, \sigma^2)$ and $\{\phi_1, \phi_2, \phi_3, \phi_4, \sigma\} = \{2.7607, -3.8106,$ $2.6535, -0.9238, 1\}$. To plot units in the simulation studies in Sections~\ref{sec:simulations} and \ref{sec:uneven}, we measure $\{X_t\}$ in metres, $t$ in seconds and assume $\Delta = 1$. Throughout, we calculate standard periodograms, corresponding to $h_t = L^{-1/2}$. Similar results with other tapers were found, with the bias in Welch's estimator being reduced but not eliminated. First, we plot in Figure~\ref{fig:ar4_peridograms} the standard (left) and debiased (right) Welch estimates with black solid lines, calculated from a single realisation from the $\textsc{ar(4)}$ model. We generate a time series of length $2^{15}$ $(32,768)$ and set $M = 64$ ($2^6$) and $L = 512$ ($2^9$) with zero overlap for both estimates. In both plots, the true spectral density is shown by the dashed grey line, and the expected value of Welch's estimator, calculated from \eqref{eqn:debiased_fft} using the auto-covariance function of the $\textsc{ar(4)}$ model, is shown by the solid grey line. Visually, the debiased Welch estimate provides a closer estimate of the true spectral density, especially in regions of low power where we know the bias to be significant.

\begin{figure}
	\centering
	\includegraphics[width = 0.65\linewidth]{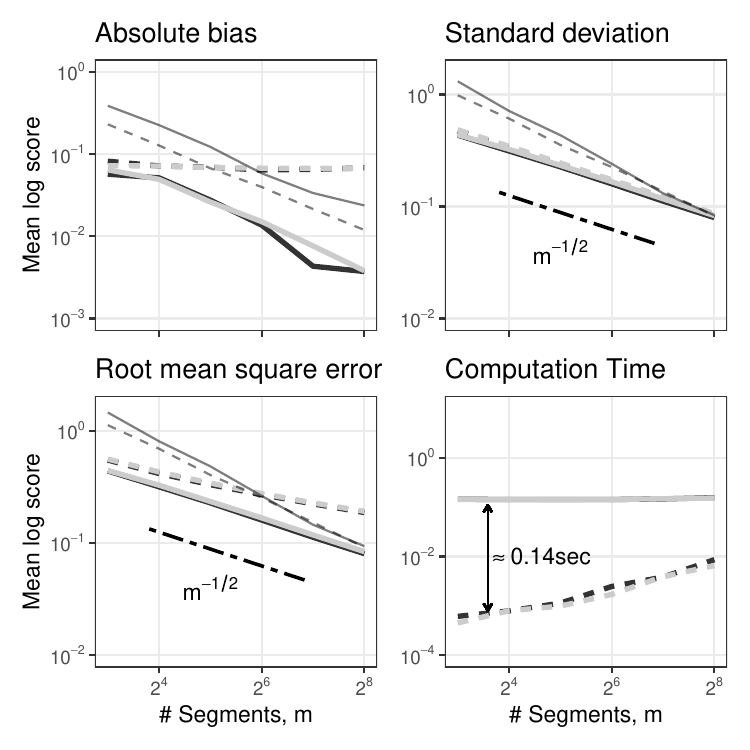}
	\caption{Empirical bias, standard deviation, root mean squared error and computation times of nonparametric estimates over $M$, for fixed $L$. Standard and debiased Welch estimates are the dashed and solid thick lines, multitaper and lag-window estimates are the dashed and solid thin lines, respectively. Each point represents the metrics calculated over an 1000 member ensemble, for each frequency, and aggregated by the mean log values. The grey lines represent an overlap of 0.5 and the black lines with an overlap of 0.}
	\label{fig:ar4_metrics}
\end{figure}

We now explore how the bias, standard deviation and root-mean-square error of the Welch estimators behave as a function of $M$ for fixed $L$. For each $M \in \{8, 16, 32, 64, 128, 256\}$ and $p \in \{0, 0.5\}$ we generate 1000 random time series from the $\textsc{ar(4)}$ model with fixed segment length $L = 1024$ ($2^{10}$). For each $M$ and $p$, we calculate, over the ensemble, the empirical mean absolute bias, standard deviation, and root-mean-square error, at each frequency. We aggregate the scores by computing the mean log value across frequency and for each metric; this makes the error process of the metrics independent of the expected power at each frequency and hence aggregates the frequencies with equal importance. Results are shown in Figure~\ref{fig:ar4_metrics}. The thick dashed lines show the results from the Welch estimator, the thick solid lines from the debiased Welch estimator, the grey lines with an overlap of $p = 0.5$ and the black lines with an overlap of $p = 0$ (i.e. Bartlett's estimator). As a point of comparison to other nonparametric spectral density estimation methodologies, we also include the results from a lag-window estimator with modified Daniell window, shown by the thin solid line, and a multitaper estimator calculated with Slepian sequence tapers, shown by the thin dashed line. We choose the window-width and number of tapers to provide similar estimator bandwidth for each value of $M$. The results of this simulation study support a number of our theoretical claims. The black dot-dashed line shows the $M^{-1/2}$ convergence rate which is consistent with the theory established in Theorem~\ref{the:welch_var}a for the Welch estimator, and is maintained for the debiased Welch estimator. The lag-window, multitaper estimator, and debiased Welch estimators all exhibit similar asymptotic behaviour in absolute bias. The absolute bias tends to zero for the debiased Welch estimator as the approximation in Theorem 2 improves with increasing $M$, justifying the computation of the debiased Welch estimates using \eqref{eqn:least_squares2}. The absolute bias does not converge to zero for the standard Welch estimator due to the finite-sample bias of periodograms of fixed length $L$ as already established in Theorem 1a. Finally, we show computation times of the standard and debiased Welch estimates. As expected, the debiased Welch estimator is approximately constant over $M$. The discrepancy in time is due to the additional operations on top of calculating Welch's estimator and is negligible in practice. We further include in the supplementary material a similar simulation study that examines the performance of the estimators as a function of $L$ for fixed $M$. We see similar benefits from using the debiased Welch estimator as demonstrated above for fixed $L$.


\subsection{Mat\'ern processes} \label{sec:matern}

\begin{figure}[b!]
	\centering
	\includegraphics[width = 0.9\linewidth]{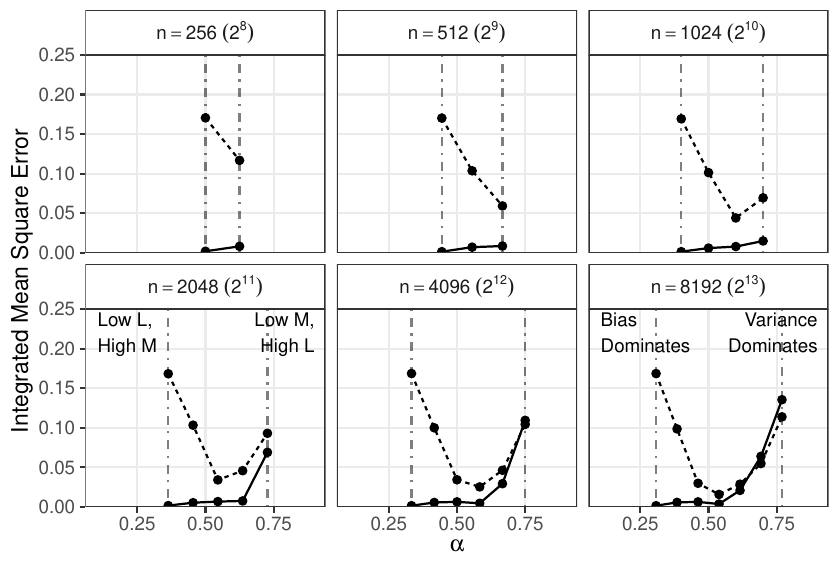}
	\caption{Integrated mean squared error for Welch (dashed) and debiased Welch (solid) estimators calculated over $\alpha$, defined by $L = n^\alpha$ with $p=0$, so that $M = n^{1 - \alpha}$, and for time-series lengths $n \in \{2^8, \dots, 2^{13}\}$. The left vertical dot-dashed lines correspond to parameterisation of $L = 16$ and the right correspond to $M = 8$.}
	\label{fig:matern_imse}
\end{figure}

Section~\ref{sec:ar4} examines performance of the debiased Welch estimator as a function of $M$, for medium to long lengths of time-series and for a discrete time process. We now analyse the debiased Welch estimator for shorter time series for a continuous time process with fixed $n$ and varying $\alpha$, defined by $L=n^\alpha$. We simulate from a Mat\'ern process parameterised as
\begin{equation*}
    \gamma(\tau) = \frac{2\sigma^2}{\Gamma(\nu)2^{\nu}} |\tau|^{\nu} \kappa_{\nu}(|\tau|), \hspace{5mm} f(\omega) = \frac{\sigma^2}{(\omega^2 + \lambda^2)^{\nu + 1/2}},
\end{equation*}
with $\sigma = 1$, $\lambda = 0.1$ and $\alpha = 4/3$, where $\Gamma(\cdot)$ denotes the gamma function and $\kappa_{\nu}(\cdot)$ the order $\nu$ modified Bessel function of the second kind. The Mat\'ern process is asymptotically log-linear in terms of its spectral decay and so is useful for modelling many observed processes; for example, background ocean spectra and turbulent dissipation in oceanography \citep{garrett1972space,lilly2017fractional}, stellar rotation periods and red-shift in astronomy \citep{foreman2017fast} and volatility processes in finance and economics \citep{heinrich2019hybrid}. Figure~\ref{fig:matern_imse} shows the integrated mean squared error averaged over 1000 random simulations of the Welch (dashed) and debiased Welch (solid) estimates as a function of $\alpha$ and for $n \in \{2^8, \dots, 2^{13}\}$. For the shorter time-series, shown in the top row, the debiased Welch estimator performs better over all values of $\alpha$. There is insufficient resolution to observe the expected bias-variance trade-off in the Welch estimator; although, with increasing $n$, shown in the bottom row, this behaviour becomes apparent. For small values of $\alpha$, bias dominates the the error; conversely, for large values of $\alpha$ variance dominates. The debiased Welch estimator removes the majority of the error in the bias dominated regions. For certain high values of $\alpha$ the debiased Welch estimator does not perform as well, although these values correspond to small $M$, a parameterisation that we do not recommend for the debiased Welch estimator. As discussed in Remark~2 following Theorem~\ref{the:variance}, we recommend scaling $\alpha < 1/3$ to minimise mean squared error: the debiased Welch estimator gives better performance across all $n$ in this region thus supporting the theoretical results.

\section{Application to coastal wave monitoring}

We now demonstrate the use of the debiased Welch estimator applied to a measured data-set of coastal sea-surface heights. The data are measured from two locations with differing mean water depths on the Southern coast of Western Australia: Ocean Beach and Torbay. Contemporaneous measurements of sea-surface height are taken from acoustic and pressure sensors. Acoustic sensors provide measurements of the true process via acoustic ranging and pressure sensors measure pressure variation on the sea-bed and translate this to a measurement of sea surface height. As compared to pressure sensors, acoustic sensors are more expensive, and due to larger power requirements, do not allow for long field deployments. Shallow water wave theory states that pressure is attenuated through depth as
\begin{equation} \label{eqn:attenuation}
    \Lambda_p(z;k) = \frac{\cosh(k(d+z))}{\cosh(k d)},
\end{equation}
where $d$ is the undisturbed depth, $-z$ is depth from the sea-surface to the sensor and, here, $k$ is wavenumber. We can also define \eqref{eqn:attenuation} as a function of frequency using the dispersion relation \citep{dean1991water}. We show two 5-minute chunks of the measured time series from both sites in Figure~\ref{fig:wave_ts}; data are captured at $1 \, \mathrm{Hz}$ such that $\Delta = 1$, and the solid grey and dashed black lines are measurements from the acoustic and pressure sensors, respectively. On the right of Figure~\ref{fig:wave_ts} we show pressure attenuation over frequency. The high frequency signal in the Torbay measurements see larger attenuation than in the Ocean Beach measurements due to the deeper mean water level. Due to the attenuation of the high frequency signal, in both records, we expect a high dynamic range in the pressure signal, and hence non-ignorable spectral bias due to blurring.

\begin{figure}[t!]
  \centering
  \includegraphics[width = \linewidth]{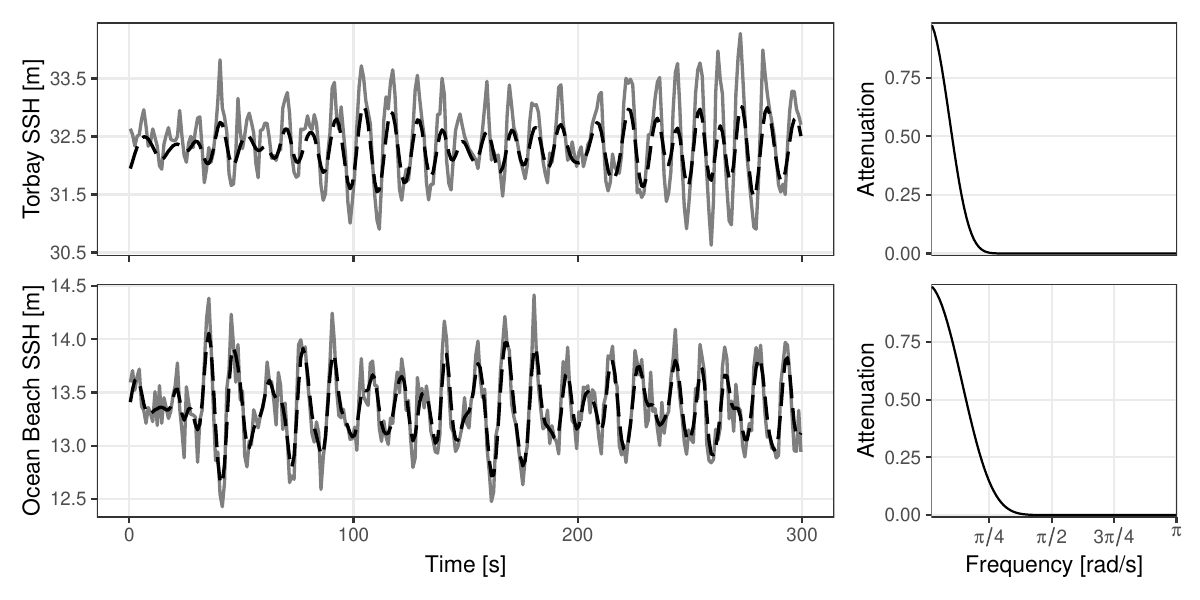}
  \caption{Contemporaneous measurements of sea-surface height from acoustic (grey) and pressure (black-dashed) sensors, measured in Western Australia at Torbay (top) and Ocean Beach (bottom). The right hand plots show the attenuation of the pressure signal over frequency and are dependent on the mean water depths.}
  \label{fig:wave_ts}
\end{figure}

We calculate power spectral density estimates of a 100 minute record ($n=6144$) of sea-surface heights at each location; the results are shown in Figure~\ref{fig:wave_data}. First, we calculate the Welch estimates, with parameters $L=256$, $M=47$ and $p=0.5$, shown in the left hand panels. Here, the grey and black lines are from the acoustic measurements, respectively. As expected, there is good agreement between the estimators at low frequencies, and the pressure signal is attenuated at higher frequencies. The spectral shape of the Ocean Beach record is interesting: there are two spectral peaks around $\pi/8\, \mathrm{rad} \, \mathrm{s}^{-1}$ and $\pi/4\, \mathrm{rad} \, \mathrm{s}^{-1}$, implying the presence of both a well formed swell (low-frequency) and wind-sea (high-frequency), likely due to a local storm. This behaviour is hard to capture parametrically and is a good example to motivate the use of non-parametric methods. The black lines in the middle plots are the debiased Welch estimates, with additional parameter $K=64$, calculated from the pressure signal, and the grey lines are the Welch estimates, for reference. For comparison, we also show the lag-windowed and multitaper estimates, described in Section~\ref{sec:ar4}. In this application, these estimates show very little bias, and so provide good validation of the performance of the debiased Welch estimator. We do not display the debiased estimate for the acoustic signal as the bias correction is very small here due to the small dynamic range. Finally, the right hand plot shows the Welch estimates of the pressure measurements inverted through the attenuation function to give an estimate of the true wave power spectra; the thick grey and black lines show the inverted Welch and debiased Welch estimates, respectively. These estimates are not shown above $3\pi/8\, \mathrm{rad} \, \mathrm{s}^{-1}$ and $5\pi/8\, \mathrm{rad} \, \mathrm{s}^{-1}$ for Torbay and Ocean Beach, respectively, as this is where the noise floor of the pressure sensor dominates the signal. The thin grey line is the acoustic Welch estimate, also shown in the left hand plots, and is used for a notion of the ground truth. Due to in-situ pre-processing of the pressure signal, whereby the signal is demeaned of tidal effects, there is some bias between the two signals at low-frequencies below $\pi/8\, \mathrm{rad} \, \mathrm{s}^{-1}$ across the entire data-set. Overall, the debiased Welch estimator gives a better estimate to the true process over a broader range of frequencies. Capturing swell dissipation, that is, the power transition from laminar/low to turbulent/high frequencies, is vital for our understanding of swell propagation. Although there is relatively low variance captured in the biased low-power regions, it is known that incorrect description of swell dissipation has large effects in forecasting models \citep[e.g.][]{ardhuin2009observation}. Here, use of the debiased Welch estimator, resolves almost twice the amount of high-frequency signal as the standard Welch estimate and thus provides a more accurate measure of swell dissipation.

\begin{figure}[t!]
\centering
\includegraphics[width = \linewidth]{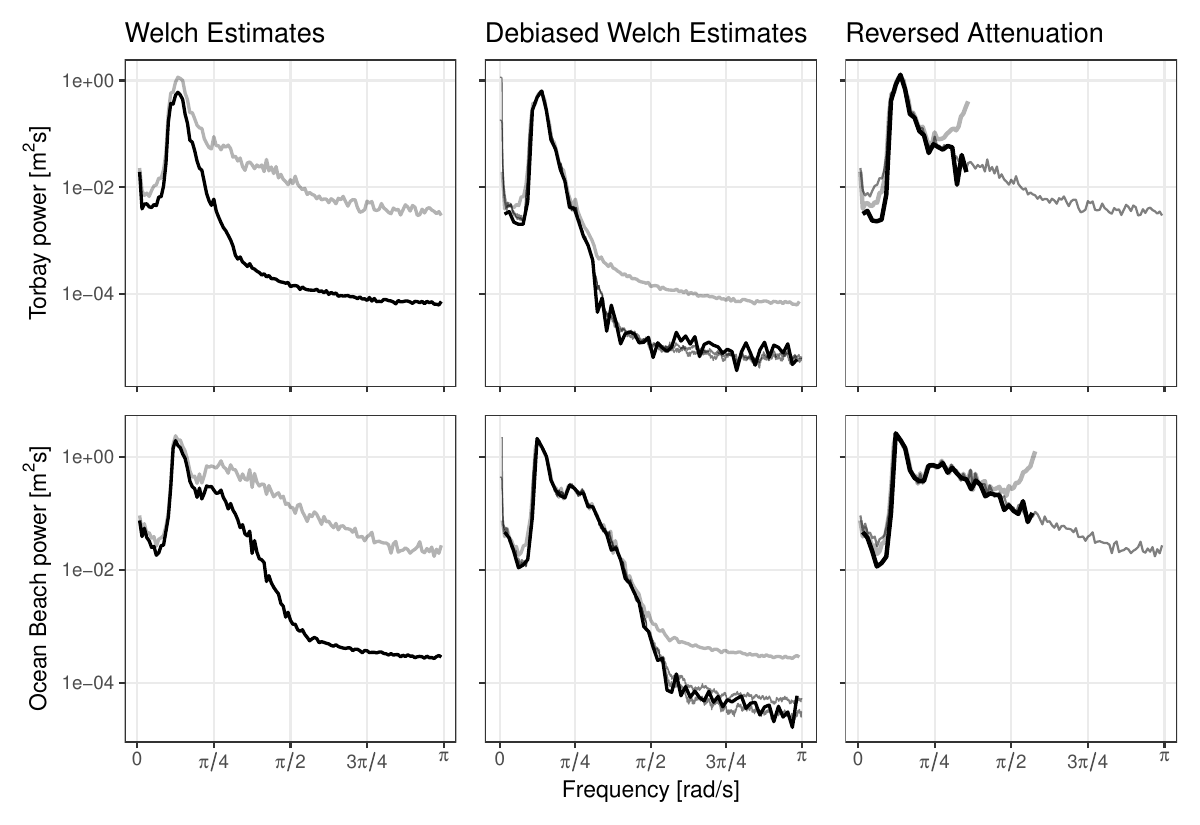}
\caption{Welch estimates of sea-surface height spectra from acoustic and pressure sensors located at Torbay (top) and Ocean Beach (bottom). The left plots show the Welch estimates of the acoustic (grey) and pressure (black) measurements. The middle plots show the debiased Welch estimates of the pressure data (black) and, for comparison, the Welch estimate (grey). The thin lines are the lag-window and multitaper estimates described in Section~\ref{sec:ar4}. The right plots show the standard (grey) and debiased (black) Welch estimates of the pressure measurements accounting for attenuation. The thin grey line is the acoustic signal used as a notion of ground truth.}
\label{fig:wave_data}
\end{figure}

\section{Unevenly spaced bases for signal compression} \label{sec:uneven}

\begin{figure}[t!]
    \centering
    \includegraphics[width = 0.77\linewidth]{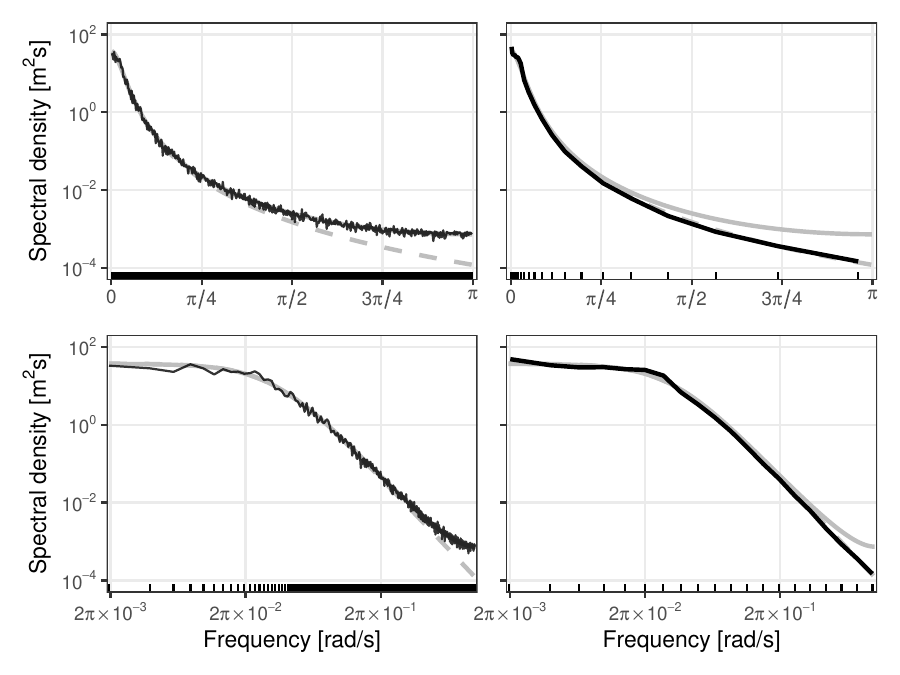}
    \caption{Welch estimates of a Mat\'ern power spectrum with uneven basis spacing. The top plots show frequency on a linear scale and the bottom plots on a log scale. The standard Welch estimate is shown left and debiased Welch estimate, with uneven basis spacing, is shown right with black solid lines. Locations of the Fourier frequencies and basis centres are shown by the rug plots on the frequency axes. In all plots, the true power spectral distribution and expected Welch estimate are shown by the dashed and solid grey lines, respectively.}
    \label{fig:materns}
\end{figure}

To be analogous to Welch's estimator, in Section~\ref{sec:debiased_welch} we define the debiased Welch estimator with even basis spacing over frequency. As noted therein, we do not require this, and can alternatively define the bases by an irregular partition over frequency. For spectra whose value varies slowly over certain neighbourhoods of frequency, we argue that a wider basis spacing would lead to little reduction of accuracy in estimating the true process. Allowing for wider, and irregularly spaced bases serves two purposes: first we achieve a degree of signal compression, as compared to the Welch estimator; and second, we reduce the variance of the estimator as there are more data informing the coefficient of each basis. The code provided alongside this manuscript accepts uneven bases by providing the $\omega^\mathrm{c}_k$ and $\delta_k$, as defined in \eqref{eqn:riemann}, as user defined inputs for all $k$. To demonstrate this concept, we return to the Mat\'ern model simulated in Section~\ref{sec:matern}. We sample a length $32,768$ $(2^{15})$ time series and calculate the Welch estimate, with parameters $M = 32$ $(2^5)$, $L = 1024$ $(2^{10})$ and no overlap, shown in the left plots of Figure~\ref{fig:materns} in black. The true power spectral density and the expected value of Welch's estimator for the Mat\'ern process are shown in the plots in Figure~\ref{fig:materns} by the dashed and solid grey lines, respectively. The Mat\'ern power spectral density is asymptotically, in $\omega$, log-linear and so we space the bases of the debiased Welch estimator evenly in log-frequency. The location of the bin-centres, $\omega_k^\mathrm{c}$, are shown by the rug-plot on the frequency axes in Figure~\ref{fig:materns}. The debiased Welch estimate, with log-spaced bases, is calculated and shown in the right plots of Figure~\ref{fig:materns} in black. The bottom plots of Figure~\ref{fig:materns} are identical to the top, except with frequency shown on a log scale. Not only do we reduce empirical bias in the estimate, we also compress the signal by a factor of $\sim50$ and reduce the standard deviation of the estimate from the true spectral density.

\section*{Acknowledgement}
All authors are supported by the ARC ITRH for Transforming energy Infrastructure through Digital Engineering (TIDE), Grant No. IH200100009. We further thank Dr. Michael Cuttler and Dr. Jeff Hanson for providing us with the wave data examined in Section 5.

\vspace*{-10pt}

\bibliographystyle{model2-names.bst}
\bibliography{paper-ref.bib}

\begin{thebibliography}{30}
\expandafter\ifx\csname natexlab\endcsname\relax\def\natexlab#1{#1}\fi
\providecommand{\url}[1]{\texttt{#1}}
\providecommand{\href}[2]{#2}
\providecommand{\path}[1]{#1}
\providecommand{\DOIprefix}{doi:}
\providecommand{\ArXivprefix}{arXiv:}
\providecommand{\URLprefix}{URL: }
\providecommand{\Pubmedprefix}{pmid:}
\providecommand{\doi}[1]{\href{http://dx.doi.org/#1}{\path{#1}}}
\providecommand{\Pubmed}[1]{\href{pmid:#1}{\path{#1}}}
\providecommand{\bibinfo}[2]{#2}
\ifx\xfnm\relax \def\xfnm[#1]{\unskip,\space#1}\fi
\bibitem[{Anderson(1971)}]{anderson1971statistical}
\bibinfo{author}{Anderson, T.W.}, \bibinfo{year}{1971}.
\newblock \bibinfo{title}{The statistical analysis of time series}.
\newblock \bibinfo{publisher}{John Wiley \& Sons}.
\bibitem[{Ardhuin et~al.(2009)Ardhuin, Chapron and
  Collard}]{ardhuin2009observation}
\bibinfo{author}{Ardhuin, F.}, \bibinfo{author}{Chapron, B.},
  \bibinfo{author}{Collard, F.}, \bibinfo{year}{2009}.
\newblock \bibinfo{title}{Observation of swell dissipation across oceans}.
\newblock \bibinfo{journal}{Geophysical Research Letters} \bibinfo{volume}{36}.
\bibitem[{Bartlett(1950)}]{bartlett1950periodogram}
\bibinfo{author}{Bartlett, M.S.}, \bibinfo{year}{1950}.
\newblock \bibinfo{title}{Periodogram analysis and continuous spectra}.
\newblock \bibinfo{journal}{Biometrika} \bibinfo{volume}{37},
  \bibinfo{pages}{1--16}.
\bibitem[{Blackman and Tukey(1958)}]{blackman1958measurement}
\bibinfo{author}{Blackman, R.B.}, \bibinfo{author}{Tukey, J.W.},
  \bibinfo{year}{1958}.
\newblock \bibinfo{title}{The measurement of power spectra from the point of
  view of communications engineering}.
\newblock \bibinfo{journal}{Bell System Technical Journal}
  \bibinfo{volume}{37}, \bibinfo{pages}{185--282}.
\bibitem[{Brillinger(1975)}]{brillinger2001time}
\bibinfo{author}{Brillinger, D.R.}, \bibinfo{year}{1975}.
\newblock \bibinfo{title}{Time series: data analysis and theory}.
\newblock \bibinfo{publisher}{Holt, Rinehart and Winston Inc.}
\bibitem[{Chui(1971)}]{chui1971concerning}
\bibinfo{author}{Chui, C.K.}, \bibinfo{year}{1971}.
\newblock \bibinfo{title}{Concerning rates of convergence of {Riemann} sums}.
\newblock \bibinfo{journal}{Journal of Approximation Theory}
  \bibinfo{volume}{4}, \bibinfo{pages}{279--287}.
\bibitem[{Dahlhaus(1985)}]{dahlhaus1985spectral}
\bibinfo{author}{Dahlhaus, R.}, \bibinfo{year}{1985}.
\newblock \bibinfo{title}{On a spectral density estimate obtained by averaging
  periodograms}.
\newblock \bibinfo{journal}{Journal of Applied Probability}
  \bibinfo{volume}{22}, \bibinfo{pages}{598--610}.
\bibitem[{Dahlhaus(1988)}]{dahlhaus1988small}
\bibinfo{author}{Dahlhaus, R.}, \bibinfo{year}{1988}.
\newblock \bibinfo{title}{Small sample effects in time series analysis: a new
  asymptotic theory and a new estimate}.
\newblock \bibinfo{journal}{The Annals of Statistics} ,
  \bibinfo{pages}{808--841}.
\bibitem[{Dean and Dalrymple(1991)}]{dean1991water}
\bibinfo{author}{Dean, R.G.}, \bibinfo{author}{Dalrymple, R.A.},
  \bibinfo{year}{1991}.
\newblock \bibinfo{title}{Water wave mechanics for engineers and scientists}.
  volume~\bibinfo{volume}{2} of \textit{\bibinfo{series}{Advanced Series on
  Ocean Engineering}}.
\newblock \bibinfo{publisher}{World Scientific Publishing Co Pte Ltd}.
\bibitem[{Fej{\'e}r(1910)}]{fejer1910lebesguessche}
\bibinfo{author}{Fej{\'e}r, L.}, \bibinfo{year}{1910}.
\newblock \bibinfo{title}{Lebesguessche {Konstanten und divergente
  Fourierreihen}.}
\newblock \bibinfo{journal}{Journal f{\"u}r die Reine und Angewandte
  Mathematik} \bibinfo{volume}{138}, \bibinfo{pages}{22--53}.
\bibitem[{Foreman-Mackey et~al.(2017)Foreman-Mackey, Agol, Ambikasaran and
  Angus}]{foreman2017fast}
\bibinfo{author}{Foreman-Mackey, D.}, \bibinfo{author}{Agol, E.},
  \bibinfo{author}{Ambikasaran, S.}, \bibinfo{author}{Angus, R.},
  \bibinfo{year}{2017}.
\newblock \bibinfo{title}{Fast and scalable {Gaussian} process modeling with
  applications to astronomical time series}.
\newblock \bibinfo{journal}{The Astronomical Journal} \bibinfo{volume}{154},
  \bibinfo{pages}{220}.
\bibitem[{Garrett and Munk(1972)}]{garrett1972space}
\bibinfo{author}{Garrett, C.}, \bibinfo{author}{Munk, W.},
  \bibinfo{year}{1972}.
\newblock \bibinfo{title}{Space-time scales of internal waves}.
\newblock \bibinfo{journal}{Geophysical Fluid Dynamics} \bibinfo{volume}{3},
  \bibinfo{pages}{225--264}.
\bibitem[{Heinrich et~al.(2019)Heinrich, Pakkanen and
  Veraart}]{heinrich2019hybrid}
\bibinfo{author}{Heinrich, C.}, \bibinfo{author}{Pakkanen, M.S.},
  \bibinfo{author}{Veraart, A.E.}, \bibinfo{year}{2019}.
\newblock \bibinfo{title}{Hybrid simulation scheme for volatility modulated
  moving average fields}.
\newblock \bibinfo{journal}{Mathematics and Computers in Simulation}
  \bibinfo{volume}{166}, \bibinfo{pages}{224--244}.
\bibitem[{Katznelson(1968)}]{katznelson1968introduction}
\bibinfo{author}{Katznelson, Y.}, \bibinfo{year}{1968}.
\newblock \bibinfo{title}{An introduction to harmonic analysis}.
\newblock \bibinfo{publisher}{John Wiley \& Sons, Inc.}
\bibitem[{Lawson and Hanson(1974)}]{lawson1995solving}
\bibinfo{author}{Lawson, C.L.}, \bibinfo{author}{Hanson, R.J.},
  \bibinfo{year}{1974}.
\newblock \bibinfo{title}{Solving least squares problems}.
\newblock \bibinfo{publisher}{Prentice Hall}.
\bibitem[{Lilly et~al.(2017)Lilly, Sykulski, Early and
  Olhede}]{lilly2017fractional}
\bibinfo{author}{Lilly, J.M.}, \bibinfo{author}{Sykulski, A.M.},
  \bibinfo{author}{Early, J.J.}, \bibinfo{author}{Olhede, S.C.},
  \bibinfo{year}{2017}.
\newblock \bibinfo{title}{{Fractional Brownian motion, the Mat{\'e}rn process,
  and stochastic modeling of turbulent dispersion}}.
\newblock \bibinfo{journal}{Nonlinear Processes in Geophysics}
  \bibinfo{volume}{24}, \bibinfo{pages}{481--514}.
\bibitem[{Nuttall(1971)}]{nuttall1971spectral}
\bibinfo{author}{Nuttall, A.H.}, \bibinfo{year}{1971}.
\newblock \bibinfo{title}{Spectral estimation by means of overlapped fast
  {Fourier} transform processing of windowed data}. volume
  \bibinfo{volume}{4169}.
\newblock \bibinfo{publisher}{Naval underwater systems center New London,
  Connecticut}.
\bibitem[{Percival et~al.(1993)Percival, Walden et~al.}]{percival1993spectral}
\bibinfo{author}{Percival, D.B.}, \bibinfo{author}{Walden, A.T.}, et~al.,
  \bibinfo{year}{1993}.
\newblock \bibinfo{title}{Spectral analysis for physical applications}.
\newblock \bibinfo{publisher}{Cambridge University Press}.
\bibitem[{Politis(2005)}]{politis2005complex}
\bibinfo{author}{Politis, D.N.}, \bibinfo{year}{2005}.
\newblock \bibinfo{title}{Complex-valued tapers}.
\newblock \bibinfo{journal}{IEEE Signal Processing Letters}
  \bibinfo{volume}{12}, \bibinfo{pages}{512--515}.
\bibitem[{Politis(2024)}]{politis2024scalable}
\bibinfo{author}{Politis, D.N.}, \bibinfo{year}{2024}.
\newblock \bibinfo{title}{Scalable subsampling: computation, aggregation and
  inference}.
\newblock \bibinfo{journal}{Biometrika} \bibinfo{volume}{111},
  \bibinfo{pages}{347--354}.
\bibitem[{Politis and Romano(1995)}]{politis1995bias}
\bibinfo{author}{Politis, D.N.}, \bibinfo{author}{Romano, J.P.},
  \bibinfo{year}{1995}.
\newblock \bibinfo{title}{Bias-corrected nonparametric spectral estimation}.
\newblock \bibinfo{journal}{Journal of time series analysis}
  \bibinfo{volume}{16}, \bibinfo{pages}{67--103}.
\bibitem[{Shumway and Stoffer(2000)}]{shumway2000time}
\bibinfo{author}{Shumway, R.H.}, \bibinfo{author}{Stoffer, D.S.},
  \bibinfo{year}{2000}.
\newblock \bibinfo{title}{Time series analysis and its applications}.
  volume~\bibinfo{volume}{3}.
\newblock \bibinfo{publisher}{Springer}.
\bibitem[{Slepian(1978)}]{slepian1978prolate}
\bibinfo{author}{Slepian, D.}, \bibinfo{year}{1978}.
\newblock \bibinfo{title}{Prolate spheroidal wave functions, {Fourier}
  analysis, and uncertainty—v: The discrete case}.
\newblock \bibinfo{journal}{Bell System Technical Journal}
  \bibinfo{volume}{57}, \bibinfo{pages}{1371--1430}.
\bibitem[{Sykulski et~al.(2019)Sykulski, Olhede, Guillaumin, Lilly and
  Early}]{sykulski2019debiased}
\bibinfo{author}{Sykulski, A.M.}, \bibinfo{author}{Olhede, S.C.},
  \bibinfo{author}{Guillaumin, A.P.}, \bibinfo{author}{Lilly, J.M.},
  \bibinfo{author}{Early, J.J.}, \bibinfo{year}{2019}.
\newblock \bibinfo{title}{The debiased {Whittle} likelihood}.
\newblock \bibinfo{journal}{Biometrika} \bibinfo{volume}{106},
  \bibinfo{pages}{251--266}.
\bibitem[{Thomson(1982)}]{thomson1982spectrum}
\bibinfo{author}{Thomson, D.J.}, \bibinfo{year}{1982}.
\newblock \bibinfo{title}{Spectrum estimation and harmonic analysis}.
\newblock \bibinfo{journal}{Proceedings of the IEEE} \bibinfo{volume}{70},
  \bibinfo{pages}{1055--1096}.
\bibitem[{Thomson et~al.(1995)Thomson, Maclennan and
  Lanzerotti}]{thomson1995propagation}
\bibinfo{author}{Thomson, D.J.}, \bibinfo{author}{Maclennan, C.G.},
  \bibinfo{author}{Lanzerotti, L.J.}, \bibinfo{year}{1995}.
\newblock \bibinfo{title}{Propagation of solar oscillations through the
  interplanetary medium}.
\newblock \bibinfo{journal}{Nature} \bibinfo{volume}{376},
  \bibinfo{pages}{139--144}.
\bibitem[{Tobar(2019)}]{tobar2019band}
\bibinfo{author}{Tobar, F.}, \bibinfo{year}{2019}.
\newblock \bibinfo{title}{Band-limited {Gaussian} processes: The sinc kernel}.
\newblock \bibinfo{journal}{Advances in Neural Information Processing Systems}
  \bibinfo{volume}{32}.
\bibitem[{Welch(1967)}]{welch1967use}
\bibinfo{author}{Welch, P.}, \bibinfo{year}{1967}.
\newblock \bibinfo{title}{The use of fast {Fourier} transform for the
  estimation of power spectra: a method based on time averaging over short,
  modified periodograms}.
\newblock \bibinfo{journal}{IEEE Transactions on audio and electroacoustics}
  \bibinfo{volume}{15}, \bibinfo{pages}{70--73}.
\bibitem[{Zhurbenko(1980)}]{zhurbenko1980efficiency}
\bibinfo{author}{Zhurbenko, I.G.}, \bibinfo{year}{1980}.
\newblock \bibinfo{title}{On the efficiency of spectral density estimates of a
  stationary process}.
\newblock \bibinfo{journal}{Theory of Probability and its Applications}
  \bibinfo{volume}{25}, \bibinfo{pages}{466--15}.
\bibitem[{Zhurbenko(1986)}]{zurbenko1986spectral}
\bibinfo{author}{Zhurbenko, I.G.}, \bibinfo{year}{1986}.
\newblock \bibinfo{title}{The spectral analysis of time series}.
\newblock \bibinfo{publisher}{Elsevier North-Holland, Inc.}

\end{thebibliography}

\end{document}